\documentclass[11pt]{article}
\usepackage{enumerate, color}
\usepackage{multido}
\usepackage{url}
\usepackage{mathrsfs}
\usepackage{mathtools}

\usepackage{amsmath,amsfonts,amsthm,amssymb, graphicx}
\swapnumbers
\theoremstyle{plain}

\newtheorem{thm}{THEOREM}[section]
\newtheorem{lm}[thm]{LEMMA}
\newtheorem{cl}[thm]{COROLLARY}
\newtheorem{prop}[thm]{PROPOSITION}
\theoremstyle{definition}

\theoremstyle{definition}
\newtheorem{remark}[thm]{Remark}
\newcommand{\upchi}{\raise1pt\hbox{$\chi$}}
\newcommand{\R}{{\mathord{\mathbb R}}}

\newcommand{\C}{{\mathord{\mathbb C}}}
\newcommand{\Z}{{\mathord{\mathbb Z}}}

\newcommand{\E}{{\mathcal{E}}}

\renewcommand{\|}{{\Vert}}
\numberwithin{equation}{section}
\def\dd{{\rm d}}

\def\T{\mathbb{T}^1}
\def\H{\mathcal{H}}
\def\LL{\mathcal{L}}

\def\K{\mathcal{K}}

\def\E{\mathcal{E}}

\begin{document}

\title{Quantitative Bounds on the Rate of Approach to  Equilibrium for some One-Dimensional
Stochastic Non-Linear Schr\"odinger Equations}

\author{Eric A. Carlen$^{1}$, J\"urg Fr\"ohlich$^{2}$, Joel Lebowitz$^{3}$, and Wei-Min Wang$^{4}$}

\maketitle

\abstract

\textit{We establish quantitative bounds on the rate of approach to equilibrium for a system with infinitely many degrees of freedom evolving according to a one-dimensional focusing nonlinear Schr\"odinger equation with diffusive forcing. Equilibrium is described by a generalized grand canonical ensemble. Our analysis also applies to the easier case of defocusing nonlinearities.}

.

\section{Introduction}

In this paper we continue our study of the focusing non-linear Schr\"odinger equation (NLS) with diffusive forcing in one dimension, extending our earlier methods \cite{LMW,CFL} to obtain \textit{quantitative} bounds on the rate of exponential relaxation to equilibrium. 

The one-dimensional deterministic NLS that we study in the following reads
\begin{equation}\label{NLS1}
i\frac{\partial }{\partial t} \phi(x,t)  = - \frac{\partial^2}{\partial x^2} \phi(x,t) + m^2\phi(x,t) - 
\lambda \big(|\phi(x,t)|^{p-2}+ \ell. o. \big)\phi(x,t)  + \kappa \|\phi\|_2^{2r-2}\phi(x,t),
\end{equation} 
where $t\in \mathbb{R}$ is time, $x$ is a point in the circle $\T$ of circumference $L$, ``$\ell.o.$'' stands for terms in 
$\vert \phi(x,t)\vert$ of order strictly lower than $p-2$ that will henceforth be neglected; ($m,\lambda,\kappa$ are positive constants, and the exponents $p$ and $r$ satisfy $p<6$ and $r > p + 2/(6-p)$). It is well known that Eq. \eqref{NLS1} is a Hamiltonian evolution equation, and, under the conditions specified here, the Gibbs measure corresponding to its Hamiltonian functional exists; (see \cite{CFL} and further discussion below).  

Equation (\ref{NLS1}) is used to describe the slowly varying envelopes of Langmuir waves in a plasma, besides various other physical phenomena.

The evolution described by Eq. (\ref{NLS1}) corresponds to the flow generated by a Hamiltonian vector field on an infinite-dimensional phase space, $\mathcal{K}$, given by the Sobolev space $H^1(\T)$. This space consists of 
complex-valued functions, $\phi$, on $\T$ with square-integrable derivative, $\phi'$, and is equipped with the norm
$${\displaystyle \|\phi\|_{H^1(\T)} =  \left( \int_{\T} |\phi'(x)|^2{\rm d}x +  
\int_{\T} |\phi(x)|^2 \text{d}x \right)^{\frac{1}{2}}.}$$ 
The phase space $\K$ can be viewed, more precisely, as the real affine space obtained by regarding the complex space $H^{1}(\T)$ as a real Hilbert space equipped with the inner product 
$$\langle \phi,\psi\rangle_\K = \Re(\langle \phi,\psi\rangle_{H^1(\T)},$$
where $\Re(z) $ denotes the real part of $z\in \C$. 
The Hamiltonian nature of the time evolution described by Eq. \eqref{NLS1} can be made manifest by equipping the algebra of bounded Fr\'{e}chet-differentiable functionals on $\mathcal{K}$ with a Poisson bracket determined by the following brackets of the complex coordinate functions:
\begin{equation}
\label{Poisson}
\lbrace \phi(x), \phi(y)\rbrace =0, \quad \lbrace \overline{\phi}(x), \overline{\phi}(y) \rbrace =0, \quad
\lbrace \phi(x), \overline{\phi}(y) \rbrace = i \delta(x-y),
\end{equation}
for arbitrary $x, y$ in $\mathbb{T}^{1}$.

The Hamiltonian functional, $H_{\lambda, \kappa}$, on $\mathcal{K}$ corresponding to Eq. \eqref{NLS1} is defined by
\begin{equation}\label{ham}
H_{\lambda,\kappa}(\phi) := \frac12  \int_{\T} (m^2|\phi(x)|^2+|\phi'(x)|^2){\rm d}x  - \frac{\lambda} {p} \int_{\T} |\phi(x)|^p {\rm d}x +  \frac{\kappa}{2r} \|\phi\|_2^{2r}\ ,
\end{equation}
where $$\lambda \in \mathbb{R}, \quad p < 6, \quad \kappa > 0 \text{  if  }\lambda > 0, \text{  and  } \kappa \geq 0 \text {  otherwise},$$
with $r> p+\frac{2}{6-p}, \text{   for  } \lambda >0,$ and $r=0$, in the defocusing case, ($\lambda < 0$). Since every function $\phi \in H^1(\T)$ is bounded and hence in $L^{p}(\mathbb{T}^{1})$, for all $p$, the Hamiltonian is well-defined and finite on all of $\K$.
Using the Poisson brackets determined by \eqref{Poisson}, one easily verifies that the NLS equation \eqref{NLS1} is equivalent to the equation
$$\dot{\phi}(x,t)= \{ H_{\lambda, \kappa}(\phi), \phi(x,t) \},$$
which renders the Hamiltonian nature of \eqref{NLS1} manifest.
The last term on the right side of \eqref{ham}, which merely gives rise to a time-dependent phase of solutions to equation (\ref{NLS1}), enforces a lower bound on the Hamiltonian $H_{\lambda,\kappa}$, for an appropriate choice of the exponent $r$ and the constant $\kappa$. This will play an important role in our considerations.\\
We remark that, in our analysis, the function $\vert \phi(x) \vert^{p}$ under the integral in \eqref{ham} could be replaced by a more general functional of $\phi(x)$ bounded by a power of $\vert \phi(x) \vert$ and also by certain non-local functionals of $\phi$.

Equation (\ref{NLS1}) can be written 
as an infinite-dimensional ordinary differential equation:
$$ \dd \phi(t) = JDH_{\lambda,\kappa}(\phi(t))\dd t \ ,$$
where $J$ is the complex structure defined by
\begin{equation}\label{Jdef}
J\phi = i \phi, \quad \text{  for an arbitrary vector   }\,\phi \text{  tangent to   } \,\mathcal{K},
\end{equation}
and $D$ denotes the Fr\'echet derivative defined on functionals on $\mathcal{K}$.

For $p< 6$, the Gibbs measure corresponding to the Hamiltonian $H_{\lambda, \kappa}$, which in standard physical notation can be written as
\begin{equation}\label{fg}
{\rm d}\mu_{\beta, \lambda, \kappa}(\phi) :=  \frac{1}{Z_{\beta, \lambda, \kappa}} e^{-\beta H_{\lambda,\kappa}(\phi)} \mathcal{D} \phi  \mathcal{D} \overline{\phi}\ ,
\end{equation}
is well defined  provided $r > p+2/(6-p)$; see \cite[Theorem 3.6]{CFL}. Henceforth we will sometimes omit the letters 
$\beta$,  $\lambda$, $\kappa$ and $r$ from our notation, writing $H$, instead of $H_{\lambda, \kappa}$, and $\text{d}\mu$, instead of $\text{d}\mu_{\beta, \lambda, \kappa}$.

The measure ${\rm d}\mu$ is absolutely continuous with respect to the Gaussian measure ${\rm d}\mu_0$ defined by
\begin{equation}\label{fg2}
{\rm d}\mu_{0} :=  \frac{1}{Z} e^{-\beta H_{0}(\phi)} \mathcal{D} \phi  \mathcal{D} \overline{\phi}\ ,
\end{equation}
where
\begin{equation}\label{h0}
H_0(\phi)  = \frac12 \int_{\T}[ |\phi'(x)|^2 + m^2|\phi(x)|^2]{\rm d}x  \ 
\end{equation}
is the free Hamiltonian with mass $m>0$. 
The covariance of the Gaussian $\text{d}\mu_{0}$ is given by the operator
\begin{equation}\label{cov}
C := \beta^{-1}(m^2 - \Delta)^{-1}\ .
\end{equation}

Let $\sigma$ be a self-adjoint Hilbert-Schmidt operator on $\K$, so that $\sigma^2$ is a positive, trace-class operator on $\K$. Let $w(t)$ denote ``Brownian motion on $\K${''}, and consider the stochastic differential equation 
\begin{equation}\label{noise2}
 \dd \phi(t) = JDH(\phi(t))\dd t  - \frac{\beta}{2} \sigma^2 DH(\phi) \dd t + \sigma \dd w(t)\ .
 \end{equation}
Associated to the stochastic differential equation (\ref{noise2}) is the Kolmogorov backward equation
${\displaystyle 
\frac{\partial}{\partial t}F = \LL F}$,
for smooth functionals $F$  on the  phase space $\mathcal{K}$, where $\LL$ is the generator of the transition function associated with the process in (\ref{noise2}); it is determined by
$${\displaystyle \frac{{\rm d}}{{\rm d} t} \mathbb{E}F(\phi_t) =    \mathbb{E}\LL F(\phi_t)},$$
with $\mathbb{E}$ denoting the expectation with respect to the law of the stochastic process. 
Using Ito's formula, one finds that
\begin{equation}\label{noise3}
\LL F(\phi) = \langle JDH(\phi),DF\rangle -  \mathcal{H} F,
\end{equation}
where $\mathcal{H}$ is the operator corresponding to the quadratic form defined by
\begin{equation}\label{noise4}
 \langle F, \mathcal{H} F\rangle_{L^2(\mu)}  =  \E(F),  \ 
 \end{equation}
with $\mathcal{E}$ given by
\begin{equation}\label{noise5}
\E(F) := \int_\Omega \langle D F,\sigma^2 D F\rangle \frac{1}{Z} e^{-\beta H(\phi)} \mathcal{D} \phi  \mathcal{D} \overline{\phi}.
\end{equation}
The positive quadartic form (``metric'') $\sigma^{2}$ appearing on the right side of  \eqref{noise5} is defined more precisely in \eqref{metric}.\\
The Kolmogorov forward equation is then 
${\displaystyle 
\frac{\partial}{\partial t}\rho = \LL^* \rho}$,
where $\LL^*$ is the adjoint of $\LL$ in the scalar product of $L^2(\mu)$, and $\rho$ is a finite measure on 
$\mathcal{K}$. In a previous study of this model (see \cite{LMW}) a cutoff on the number of modes in the fields was introduced, and the existence of a strictly positive spectral gap for the finite-dimensional problem with cutoffs was proven.

Equation \eqref{NLS1} with $p=4$ has been studied in detail in \cite{CFL}, where it is shown that the semigroup 
$(e^{t\LL})_{t\geq 0}$ generated by the operator $\mathcal{L}$ in \eqref{noise3} is ergodic, 
and that $\LL$ has a strictly positive spectral gap above its lowest eigenvalue, 
provided $r>9$  and $\sigma$ is chosen to be a fractional power of the covariance $C$ 
introduced in (\ref{cov}):
\begin{equation}\label{metric}
\sigma^{2} = C^s, \quad{\rm with}\quad \frac{7}{8} < s < 1\ .
\end{equation}
Under these conditions on $r$ and $s$, a certain operator arising in the analysis of the Dirichlet form can be shown to be 
trace-class, and this provides the crucial compactness property  that is used in \cite{CFL} to prove 
the existence of a spectral gap,
for all positive   values of $\lambda$ and $\kappa$.   Because the proof in \cite{CFL} only exploits the compactness of a certain operator, it does not yield quantitative information on the size of the specrtal gap.
In the present work we prove \textit{quantitative bounds on the gap above the ground state energy} in the spectrum of 
$\mathcal{L}$, for \textit{all} 
values of $\lambda$ and 
$\kappa>0$. We will actually prove a quantitative logarithmic Sobolev inequality, for all $\lambda$ and $\kappa>0$, 
which implies the strict positivity of and an explicit bound on the spectral gap.  Moreover, we avoid introducing any 
cutoffs and work directly with the infinite-dimensional theory.

Note that the stochasticity in (\ref{noise2}) acts on \textit{all} phase space variables, that is, on 
the ``position variables{''} ($\Re(\phi)$) as well as the ``momentum variables{''} ($\Im(\phi)$). This is 
different from what is often studied in stochastic particle systems, where the noise typically acts only on 
the ``momentum variables{''} corresponding, in our case, to the imaginary part of $\phi$. It would be more difficult 
to prove bounds on the rate of approach to equilibrium in this case; see Section 2.

\section{Log-concave measures and logarithmic Sobolev inequalities}

In finite dimensions, the Bakry-Emery Theorem establishes a very useful link between logarithmic Sobolev inequalities 
and log-concavity of measures. We recall some relevant facts before turning to results in infinitely many dimensions.

Let $\nu$ be a finite Borel measure on $\R^n$ of the form ${\rm d}\nu = e^{-V(x)}{\rm d}x$.   
The measure $\nu$  is {\em log-concave} in case $V$ is a convex function on $\R^n$.   
For $c\in \R$,  the measure $\nu$  is {\em  $c$-log-concave} iff the Hessian of $V$, 
${\rm Hess}_V(x)$, satisfies 
\begin{equation}\label{logconcave2}
{\rm Hess}_V(x) \geq c I\ , \quad \forall x \in \mathbb{R}^{n},
\end{equation}
where $I$ is the $n\times n$ identity matrix.   
Equivalently, $\nu$ is $c$-log-concave in case $e^{c|x|^2/2}{\rm d}\nu$ is  log-concave.

Bakry and Emery proved that if
$\nu$  is   $c$-log-concave, for $c>0$, the logarithmic Sobolev inequality (with constant $c$)
\begin{equation}\label{logconcave3}
\int_{\R^n}  |f(x)|^2 \log |f(x)|^2 {\rm d}\nu(x) \leq \frac{2}{c}\int_{\R^n} |\nabla f(x)|^2(x){\rm d}\nu(x)
\ 
\end{equation}
holds for all continuously differentiable functions $f$ on $\R^n$, with $\int_{\R^n}f^2(x){\rm d}\nu(x) =1$. \\
For a differentiable function $u\in L^2(\nu)$ satisfying $\int_{\R^n}u {\rm d}\nu =0$ and $\int_{\R^n}|u|^2 {\rm d}\nu =1$, 
we set $f := \sqrt{1 -\epsilon^2} + \epsilon u$. 
For this choice of $f$ in (\ref{logconcave3}), and keeping only the leading terms in 
$\epsilon$ on both sides of \eqref{logconcave3}, one concludes that
\begin{equation}\label{logconcave4}
\int_{\R^n} |u(x)|^2 {\rm d}\nu(x) \leq \frac{1}{c}\int_{\R^n} |\nabla u(x)|^2{\rm d}\nu(x).\ 
\end{equation}
Thus, the logarithmic Sobolev inequality (\ref{logconcave3}) implies the Poincar\'e
inequality (\ref{logconcave4}), and hence positivity of the spectral gap, for the operator corresponding to the quadratic form
$\mathcal{E}(u):= \int_{\mathbb{R}^{n}} \vert \nabla u(x) \vert^{2} \text{d}\nu(x)$

Bakry and Emery proved their theorem by taking two derivatives of the relative entropy along the 
flow of the semigroup generated by the Dirichlet form. While it is likely that one could extend their analysis 
to the infinite-dimensional setting, we do not know of a suitable reference.  

There is however another approach to the Bakry-Emery Theorem relying on a theorem of 
Caffarelli that has been extended to a suitable  infinite-dimensional setting in a series of papers by 
Feyel and \"Ust\"unel \cite{FU1,FU2,FU3,FU4}.  Their results concern pairs of Dirichlet forms of the type
\begin{equation}\label{H1dir}
\E_1(F) := \int_\Omega \langle D F,\sigma^2 D F\rangle \frac{1}{Z_1} 
e^{-\beta H_{1}(\phi)} \mathcal{D} \phi  \mathcal{D} \overline{\phi}. 
\end{equation}
 and
\begin{equation}\label{H2dir}
\E_2(F) := \int_\Omega \langle D F,\sigma^2 D F\rangle \frac{1}{Z_2}
 e^{-\beta H_{2}(\phi)} \mathcal{D} \phi  \mathcal{D} \overline{\phi},
\end{equation}
where $H_1$ and $H_2$  are  Hamiltonians with the property that the probability measures appearing in the two 
Dirichlet forms are both  absolutely continuous with respect to the \textit{same} Wiener measure.   
Then, roughly speaking, if $H_{2}$ is more convex than $H_{1}$, and if the Dirichlet 
form $\E_1$ satisfies the logarithmic Sobolev inequality with constant $c$, then the 
Dirichlet form $\E_2$ satisfies the logarithmic Sobolev inequality with the {\em same} constant $c$. 
\cite{FU1,FU2,FU3,FU4}.

In our application of this result we shall take  $H_{1}$ to be a positive multiple of the quadratic free Hamiltonian \eqref{h0},
for a strictly positive mass $m$.  
It is well known, going back to results of E. Nelson, P. Federbush and L. Gross \cite{N73,F69,G}, that  the Dirichlet form associated with the corresponding Gaussian measure satisfies the logarithmic Sobolev inequality with an explicitly computable, sharp constant.  Thus, all that is required to prove an explicit logarithmic Sobolev inequality for the Dirichlet form (\ref{noise5}) is to prove that $H$ is more convex than some strictly positive multiple of $H_0$. This turns out to be true for sufficiently small values of $\lambda$, cf. sect. 5.2 in the Appendix. For large values of 
$\lambda$,  $H$ fails to be convex. However, the failure of convexity only occurs in finitely many low-energy modes. 
For \textit{all} values of $\lambda$, we will therefore be able to find a function $W$ that depends on $\phi$ only 
through finitely many modes such that the functional $H + W$ is more convex than a strictly positive multiple of $H_0$, 
and moreover, we shall do this with a point-wise bounded perturbation $W$.   This allows us to apply another 
theorem on logarithmic Sobolev inequalities for a pair of Dirichlet forms such as \eqref{H1dir} and \eqref{H2dir}, but this time with
$H_2 = H_1 + W$ with $\|W\|_\infty < \infty$. The Holley-Stroock Lemma \cite{HS} then says that if $\E_1$ satisfies a logarithmic Sobolev
inequality with a constant $c$ as in \eqref{logconcave3} then $\E_2$ satisfies a logarithmic Sobolev inequality with a constant that 
is no smaller than $ce^{-2\|W\|_\infty}$. Then, as in the passage from \eqref{logconcave3} to \eqref{logconcave4}, we obtain a spectral gap by linearizing around the constant function. Note that while a Dirichlet form may satisfy a spectral gap inequality  without satisfying a logarithmic Sobolev inequality, one advantage of working with logarithmic Sobolev inequalities when they hold is that \eqref{logconcave3}
can be written as
$$
\int_{\R^n}  |f(x)|^2 \log |f(x)|^2 {\rm d}\nu(x)  + \|f\|_2^2 \log \|f\|_2^2 \leq \frac{2}{c}\int_{\R^n} |\nabla f(x)|^2(x){\rm d}\nu(x)
$$
valid for all $f\in L^2$ without any orthogonality constraint such as one has in the spectral gap inequality 
\eqref{logconcave4}. This absence of an orthogonality  constraint, which is quite sensitive to bounded changes of measure, gives the logarithmic Sobolev inequality an advantageous quality of robustness. 
 Our main result is the following theorem.
 
 \begin{thm} Let $H$ be the Hamiltonian specified in Eq. \eqref{ham}, with $p=4$ and $r>5$. Let $\E$ be the Dirichlet form introduced in \eqref{noise5}, and let $\E_0$ be the ``Gaussian Dirichlet form'' given by the same formula, with $H_0$ in place of $H$.  Let $C_0$ denote the constant appearing in the  logarithmic Sobolev inequality for $\E_0$,
 $$\int |F|^2\log |F|^2 {\rm d}\mu_0 \leq \frac{2}{C_0}\E_0(F,F),$$
  for all $F$ with $\int |F|^2{\rm d}\mu_0 =1$.\\
Then, for all $r > 5$ and all positive values of $\lambda$ and $\kappa$, there is a computable constant $C$ depending on these parameters such that the Dirichlet form $\E$ satisfies 
 $$\int |F|^2\log |F|^2 {\rm d}\mu_0 \leq \frac{2}{C}\E(F,F), $$
 for all $F$ with $\int |F|^2{\rm d}\mu =1$.   As $\lambda$ increases to infinity, the constant $C$ diverges to infinity exponentially in a power of $\lambda$. This power is always at least $2$, and approaches $2$ as $r$ approaches infinity. 
 \end{thm}

The Holley-Stroock Lemma has been used for related models by Gordon Blower  \cite{Bl}; see also \cite{BBD}. Combining this Theorem with the  results of Caffarelli, Feyel and  \"Ust\"unel we are able to carry out a convexity comparison directly in the infinite-dimensional setting and to avoid sharp cut-offs or finite-dimensional approximations.

\section{Convexity comparison}

In this section we estimate the Hessians of the various terms in the Hamiltonian $H_{\lambda,\kappa}$. The term that has the potential to spoil the convexity is the interaction term $-\frac{\lambda}{p}\|\phi\|_p^p$, which is concave.   To avoid complicated remainder terms, we specialize to the case $p=4$ and define 
\begin{equation}\label{V1forf}
V_1(\phi) = \frac{1}{4}\int_{\T }|\phi(x)|^4{\rm d}x\ .
\end{equation}
Given two complex numbers, $z$ and $w$, we let $\theta\in [0,2\pi)$ be such that $\Re{\overline{z}w} = \cos(\theta)|z||w|$.  The function $t\mapsto t^{2}$ is convex on $[0,\infty)$, and hence $t\mapsto (|z|^2+|w|^2+ 2|z||w|t)^{2}+ (|z|^2+|w|^2- 2|z||w|t)^{2}$ is an increasing function of $t$ on $[0,1]$. Therefore,
$$|z+w|^{4} + |z-w|^{4} \leq ||z|+|w||^4 +  ||z|-|w||^4\ .$$
It follows from this inequality and the fact that $V_1$ is convex that 
$$0 \leq \frac12[V_1(\phi+\eta) + V_1(\phi-\eta)]  - V_1(\phi) \leq  \frac12[V_1(|\phi|+|\eta|) + V_1(|\phi|-|\eta|)]  - V_1(\phi)\ ,
$$
and
\begin{eqnarray}\label{V1dif}
\frac12[V_1(|\phi|+|\eta|) + V_1(|\phi|-|\eta|)]  - V_1(\phi)  &=& \int_{\T}  (3|\phi|^2|\eta|^2  + \frac14|\eta|^4){\rm d}x  \\
&=&   3\int_{\T}|\phi|^2|\eta|^2{\rm d}x  + \frac14 \|\eta\|_4^4\ .
\end{eqnarray}
Therefore, for any Hilbert space $\H\subset L^2$ with the property that $ \|\eta\|_4^4 = o(\|\eta\|_\H^2)$, as $\Vert \eta \Vert_{\mathcal{H}} \searrow 0$, the Hessian of $V_1$ at $\phi\in \H$,
${\rm Hess}_{V_1}(\phi)$, satisfies
\begin{equation}\label{bound00}
0 \leq \langle \eta, {\rm Hess}_{V_1}(\phi) \eta\rangle_\H  \leq  3\int_{\T} |\phi(x)|^2|\eta(x)|^2{\rm d}x\ .
\end{equation}
Note that
\begin{equation}\label{bound0}
\int_{\T}  |\phi|^2|\eta|^2{\rm d}x  \leq \|\phi\|_2^2\|\eta\|_\infty^2\ .
\end{equation}
We shall estimate $\|\eta\|_\infty$ in terms of $H_0(\eta) = m^2\|\eta\|_2^2 + \|\eta'\|_2^2$. We must, however, retain a piece of the term $m^2\|\eta\|_2^2$ in $H_0(\eta)$ for later use. Therefore, for $a>0$, we define an operator $A_a$ as 
\begin{equation}\label{Adef}
A _a := \frac{a^2}{L^2} - \Delta\ .
\end{equation}
Then
\begin{equation}\label{Adef2}
H_0(\eta) = \langle \eta, \left(m^2-\frac{a^2}{L^2}\right)\eta\rangle  + \langle \eta,A_a \eta \rangle\ .
\end{equation}

To control $\|\eta\|_\infty$, we use the following simple Sobolev embedding lemma:

\begin{lm}[Sobolev Embedding]\label{SobEm} For all $a> 0$ and all $\gamma > 1/4$, there is a universal constant $C_{a,\gamma}$  such that, for all 
functions $\psi$ on the torus in the domain of the operator $(-\Delta)^{\gamma}$,
\begin{equation}\label{twointerp2A}
\|\psi\|_\infty \leq C_{a,\gamma} L^{2\gamma-1/2}
\|A_a^\gamma \psi\|_2 \ .
\end{equation}
\end{lm}

\begin{proof}  We write $\psi(x)$ as a Fourier series:
$$\psi(x) = L^{-1/2}\sum_{k\in \Z} \widehat\psi(k) e^{2\pi i kx/L} = 
L^{2\gamma-1/2}\sum_{k\in \Z} \widehat\psi(k)\left(\frac{a^2+ (2\pi k)^2}{L^2}\right)^\gamma e^{2\pi i kx/L} (a^2+(2\pi k)^2)^{-\gamma} $$
Applying the Cauchy-Schwarz inequality yields (\ref{twointerp2A}) with 
$$C^2_{a,\gamma} = \sum_{k\in \Z} (a^2+(2\pi k)^2)^{-2\gamma}\ .$$
\end{proof}

We define  $P_n$ to  be the projector onto the span of the functions $\{ e^{-i2\pi kx/L}\ :\  -n \leq k \leq n\}$ in $L^{2}$. In what follows a decomposition into low-frequency  and high-frequency modes is crucial. Since $P_n$ commutes with  any power of 
$A_a$, we have that
\begin{equation}\label{nsplit}
\|A_a^\gamma \psi\|_2^2 = \|A_a^\gamma P_n\psi\|_2^2 + \|A_a^\gamma P_n^\perp \psi\|_2^2 \ .
\end{equation}
The next lemma is the key to much of what follows afterwards.

\begin{lm}\label{split}  For all $\psi \in H_1(\T)$,  all $a>0$, $\gamma > 1/4$ and $\epsilon>0$ such that $\gamma+\epsilon< 1/2$, 
and all $n\in \mathbb{N}$,
\begin{equation}\label{twointerp2}
\|\psi\|_\infty^2  \leq C_{a,\gamma} L^{4\gamma-1} \left(\|P_n\psi\|_2^{2-4\gamma}\|P_nA_a^{1/2}\psi\|_2^{4\gamma}
+ \frac{1}{(2\pi n/L)^{4\epsilon}}
\|P_n^\perp\psi\|_2^{2-4(\gamma+ \epsilon)}\|P_n^\perp A_a^{1/2}\psi\|_2^{4(\gamma+\epsilon)} \right) 
\end{equation}
where $C_{a,\gamma}$ is the constant specified in Lemma~\ref{SobEm}.
\end{lm}
We set
\begin{equation}\label{bound00A}
S_1(\eta) := 3C_{a,\gamma}^2 L^{4\gamma-1} \|P_n\eta\|_2^{2-4\gamma}\|P_n A_a^{1/2}\eta\|_2^{4\gamma}
\end{equation}
and 
\begin{equation}\label{bound00B}
S_2(\eta) := 3C_{a,\gamma}^2L^{4\gamma -1}  \frac{1}{(2\pi n/L)^{4\epsilon}}
\|P_n^\perp\eta\|_2^{2-4(\gamma+ \epsilon)}\|P_n^\perp A_a^{1/2}\eta\|_2^{4(\gamma+\epsilon)} \ .
\end{equation}
Combining Lemma~\ref{split} with (\ref{bound00}) and (\ref{bound0}), we obtain the bound
\begin{equation}\label{bound00B}
0 \leq \langle \eta, {\rm Hess}_{V_1}(\phi) \eta\rangle_\H  \leq  \|\phi\|_2^2 S_1(\eta)  + \|\phi\|_2^2 S_2(\eta) \ .
\end{equation}
The merit of this bound is  that the exponents of the derivative terms in $S_1(\eta)$ and $S_2(\eta)$, $\|P_n A_a^{1/2}\eta\|_2$ and 
$\|P_n^\perp A_a^{1/2}\eta\|_2$, respectively, are both less than two, allowing one to control these terms with the help of the contribution from $H_0(\eta)$. 
Moreover,  by choosing $n$ sufficiently large, one can make the constant factor 
${\displaystyle  \frac{3C_{a,\gamma}}{(2\pi n/L)^{4\epsilon}}}$ as small as one may wish, while $S_1(\phi)$ depends on $\phi$ only through finitely many modes. We shall exploit this fact to quantitatively bound the log-Sobolev constant, and hence the spectral gap, for arbitrarily large values of the coupling constant $\lambda$.

\begin{proof}[Proof of Lemma~\ref{split}] By (\ref{twointerp2A})
$$\|\psi\|_\infty^2  \leq C_{a,\gamma}^2 L^{4\gamma-1}\|A_a^\gamma \psi\|_2^2 = 
C_{a,\gamma}^2 L^{4\gamma-1} \langle \psi, A_a^{2\gamma} \psi\rangle\ .$$
Since 
$$((a^2+ (2\pi k)^2)/L^2)^{2\gamma} =  (t^{1/(1-2\gamma)})^{1-2\gamma} (t^{-1/2\gamma}(a^2+ (2\pi k)^2)/L^2)^{2\gamma}, $$
the arithmetic-geometric mean inequality yields
$$\|A_a^\gamma \psi\|_2^2  \leq  (1-2\gamma)t^{1/(1-2\gamma)} \|\psi\|_2^2 + 2\gamma t^{-1/2\gamma}\|A_a^{1/2}\psi\|_2^2\ .$$
Choosing $t$ to minimize the right side, we obtain
the interpolation inequality
\begin{equation}\label{twointerp2B}
\|A_a^\gamma \psi\|^2 \leq \|\psi\|_2^{2-4\gamma}\|A_a^{1/2}\psi\|_2^{4\gamma}\ .
\end{equation}
Applying this inequality to each of the two terms on the right side of (\ref{nsplit}) yields
\begin{equation}\label{twointerp}
\|A_a^\gamma \psi\|^2 \leq \|P_n\psi\|_2^{2-4\gamma}\|P_nA_a^{1/2}\psi\|_2^{4\gamma}
+ \|P_n^\perp\psi\|_2^{2-4\gamma}\|P_n^\perp A_a^{1/2}\psi\|_2^{4\gamma}\ .
\end{equation}
Combining \eqref{twointerp} with (\ref{twointerp2A}), we obtain that
\begin{equation}\label{twointerp2}
\|\psi\|_\infty^2  \leq C_{a,\gamma}  \left(\|P_n\psi\|_2^{2-4\gamma}\|P_nA_a^{1/2}\psi\|_2^{4\gamma}
+ \|P_n^\perp\psi\|_2^{2-4\gamma}\|P_n^\perp A_a^{1/2}\psi\|_2^{4\gamma}\right) \ .
\end{equation}
Since
${\displaystyle \|P_n^\perp A_a^{1/2}\psi\|_2^2 \geq \frac{1}{(2\pi n/L)^2} \|P_n^\perp\psi\|_2^2}$, 
\begin{equation*}
\|P_n^\perp A_a^{1/2}\psi\|_2^{4\gamma} \leq \frac{1}{(2\pi n/L)^{4\epsilon}}
\|P_n^\perp\psi\|_2^{-4\epsilon}\|P_n^\perp A_a^{1/2}\psi\|_2^{4(\gamma+\epsilon)} ,
\end{equation*}
and combining this bound with (\ref{twointerp2}) completes the proof. 
\end{proof}

The remaining terms in the Hamiltonian $H_{\lambda,\kappa}$ are much simpler to treat.  For $r\geq 1$, we define
\begin{equation}\label{V2forf}
V_2(\phi) = \frac{1}{2r}\|\phi\|_{2}^{2r}\ .
\end{equation}

\begin{lm}\label{nosplit}
\begin{equation}\label{V2bnd}
\frac12  [V_2(\phi+\eta) + V_2(\phi-\eta)] - V_2(\phi) \geq  \|\phi\|_2^{2r-2}\|\eta\|_2^2\ .
\end{equation}
\end{lm}

\begin{proof}  By the convexity of the $r$th power, for $r\geq 1$, and the parallelogram law, 
$$\frac12 [  \left( \|\phi+\eta\|_2^2 \right)^r +  \left(\|\phi-\eta \|_2^2 \right)^r] \geq 
 \left( \frac12 \left[\|\phi+\eta\|_2^2  + \|\phi-\eta\|_2^2\right]\right)^r = 
 \left(  \|\phi\|_2^2  + \|\eta\|_2^2\right)^r\ .$$
Applying the inequality $f(t+s) \geq f(s) + f'(s)t$, valid for any differentiable convex function, to the function $f(t) = t^p$, we conclude that
$$\left(  \|\phi\|_2^2  + \|\eta\|_2^2\right)^r \geq   \|\phi\|_2^{2r} +  r \|\phi\|_2^{2r-2}  \|\eta\|_2^2\ ,$$
which completes the proof.
\end{proof}

The only remaining term in the Hamiltonian $H_{\lambda,\kappa}$  is the free Hamiltonian, \mbox{$H_0(\phi) = \langle \phi, (m^2-\Delta)\phi\rangle$}, which is quadratic in $\phi$ and positive. Hence, by the  parallelogram law and the definition of $A_a$, \eqref{Adef},
\begin{eqnarray}\label{freecon}
\frac12[ H_{0}(\phi+\eta) +  H_{0}(\phi-\eta)] - H_{0}(\phi)  &=& H_0(\eta)\nonumber\\
&=&  \langle \eta , (m^2-\Delta)\eta\rangle\nonumber\\
&=& \left(m^2 - \frac{a^2}{L^2}\right)\|\eta\|_2^2 + \|A_a^{1/2}\eta\|_2^2 \nonumber\\
&=& m_a^2\|\eta\|_2^2 +   \|A_a^{1/2}\eta\|_2^2\ ,
\end{eqnarray}
where
\begin{equation}\label{madef}
m^2_a :=  m^2 - \frac{a^2}{L^2}\ .
\end{equation}

Combing the estimates in (\ref{bound00B}), (\ref{V2bnd}) and (\ref{freecon}), we obtain that
\begin{equation}\label{bound3}
\langle \eta, {\rm Hess}_{H_{\lambda,\kappa}}(\phi)  \eta\rangle_\H  \geq
m^2\|\eta\|_2^2+ \|\eta'\|_2^2  - \lambda \|\phi\|^2 [S_1(\eta) + S_2(\eta)] + \kappa \|\phi\|_2^{2r-2}\|\eta\|_2^2\ .
\end{equation}
Therefore, 
for any $\alpha \in (0,1)$,
$$\langle \eta, {\rm Hess}_{H_{\lambda,\kappa}}(\phi) \eta\rangle_\H   - \alpha \langle \eta, {\rm Hess}_{H_0}(\phi)\eta\rangle_\H$$ is bounded below by the sum of
\begin{equation}\label{hessc1}
(1-\alpha) [m^2_a\|P_n\eta\|_2^2 + \|P_n A_a^{1/2}\eta\|_2^2]  - \lambda \|\phi\|^2 S_1(\eta) + \kappa \|\phi\|_2^{2r-2}\|P_n\eta\|_2^2
\end{equation}
and 
\begin{equation}\label{hessc2}
(1-\alpha) [m^2_a\|P_n^\perp\eta\|_2^2+ \|P_n^\perp A_a^{1/2}\eta\|_2^2]  - \lambda \|\phi\|^2  S_2(\eta) + \kappa \|\phi\|_2^{2r-2}\|P_n^\perp\eta\|_2^2\,,
\end{equation}
which we estimate separately, beginning with (\ref{hessc1}). We choose $\alpha\in (0,1)$, and we define
 $t :=  \|P_n A_a^{1/2}\eta\|_2$ and $M := 3\lambda C^2_{a,\gamma} L^{4\gamma-1}\|\phi\|_2^2  \|P_n\eta\|_2^{2-4\gamma}$. We then have that
$$(1-\alpha)\|P_n A_a^{1/2}\eta\|_2^2  - \lambda \|\phi\|^2 S_1(\eta)   = (1-\alpha)t^2 - M  t^{4\gamma}\ .$$
Simple computations show
that  there is a constant $c_{\gamma,\alpha}$ depending only on $\alpha$ and $\gamma$ such that
$$(1-\alpha)t^2 - M t^{4\gamma} \geq -c_{\gamma,\alpha} M^{1/(1-2\gamma)}\,, \quad \forall t>0.$$
Using this inequality to eliminate $\|P_n A_a^{1/2}\eta\|_2^2$,
we obtain the following lower bound on the quantity in (\ref{hessc1}):
\begin{equation}\label{focus}
\left(
(1-\alpha) m^2_a      - \left( 3\lambda C_{a,\gamma}^2 L^{4\gamma-1}\right)^{1/(1-2\gamma)} 
\|\phi\|_2^{2/(1-2\gamma)}
+ \kappa r \|\phi\|_2^{2r-2} \right) \|P_n\eta\|_2^2\ .
\end{equation} 
For $r>  1+ 1/(1-2\gamma)$, let $s = r -(1/(1-2\gamma)) -1$.  Then, setting $t = \|\phi\|_2^{2/(1-2\gamma)}$,
we may write our lower bound as
\begin{equation}\label{lb}
\|\eta\|^2\left( (1-\alpha)m_a^2  - \left( 3\lambda C_{a,\gamma}^2 L^{4\gamma-1}\right)^{1/(1-2\gamma)}t+ \kappa r t^{(1-2\gamma)(r-1)} \right)\ .
\end{equation}
Recall that, below \eqref{fg}, we imposed the restriction $r > p+2/(6-p)$, which, for $p=4$, is implied by
$r > 5$.
We suppose that $(1-2\gamma)(r-1) > 1$, and, since $\gamma > 1/4$, this requires $\gamma$ to be very close to $1/4$ if $r$ is close to $5$; and,  no matter how large $r$ is, we require $\gamma < 1/2$. With $\gamma$  
 chosen as required, we define $q := (1-2\gamma)(r-1) -1$. For $b,c>0$, we have that
$$-ct + bt^{1+q} \geq -\frac{q}{1+q}\left(\frac{1}{(1+q)b}\right)^{1/q}c^{(q+1)/q}\ .$$
Setting
$$b :=\kappa r \quad \text{and   } \quad c := \left(3\lambda C_{a,\gamma}^2 L^{4\gamma-1}\right)^{1/(1-2\gamma)},$$
this inequality shows that the quantity in (\ref{lb}) is non-negative, provided that 
\begin{equation}\label{cchoice}
(1-\alpha)m_a^2   -\frac{q}{1+q}\left(\frac{1}{(1+q)\kappa r }\right)^{1/q}\left( 3\lambda C_{a,\gamma}^2 L^{4\gamma-1}\right)^{(q+1)/q(1-2\gamma)}\ 
\end{equation}
is non-negative, which is evidently satisfied if $\lambda$ is sufficiently small  or $\kappa$ is sufficiently large -- but only in these cases!  Note that the exponent $(q+1)/q(1-2\gamma)$ is at least as large as $2$, which it approaches when $r \uparrow \infty$ and $\gamma\downarrow 1/2$.

The situation is much better for the high-frequency modes. 
The same analysis 
shows that if $(1-2(\gamma+\epsilon))(r-1) > 1$, and for $q'$ defined by  
$q' := (1-2(\gamma+\epsilon))(r-1) -1$,
the quantity in (\ref{hessc2}) is non-negative, provided that
\begin{equation}\label{nchoice}
(1-\alpha)m_a^2   -\frac{q'}{1+q'}\left(\frac{1}{(1+q')\kappa r }\right)^{1/q'}
\left( 3\lambda C_{a,\gamma}^2L^{4\gamma -1}  \frac{1}{(2\pi n/L)^{4\epsilon}} 
\right)^{(q'+1)/q'(1-2(\gamma+\epsilon))} \geq 0\ .
\end{equation}
The exponent $(q'+1)/q'(1-2(\gamma+\epsilon))$ is always at least as large as $2$, which it approaches 
when $r \uparrow \infty$,  $\gamma\downarrow 1/2$ and $\epsilon\downarrow 0$.  

No matter how large $\lambda$ is or how small $\kappa$ is, the negative term can be made arbitrarily small by choosing $n$ sufficiently large. 
Thus, no matter how large the value of the coupling constant $\lambda$ may be, or how small $\kappa$ may be, there exists a finite $n\in \mathbb{N}$ such that the quantity in 
(\ref{hessc2}) is non-negative. 
For such a value of $n$, the failure of convexity only concerns the $2n+1$ lowest frequency modes. We may then compensate this failure by adding a {\em uniformly bounded} term, $W(\phi)$, to 
$H(\phi)$ that depends on
$\phi$ only through the $2n+1$ lowest-frequency modes, with the property that the Dirichlet form associated with the perturbed measure
$$\frac{1}{Z} e^{-\beta [H(\phi)+ W(\phi)]}\mathcal{D} \phi  \mathcal{D} \overline{\phi}$$
satisfies a logarithmic Sobolev inequality. As explained in the last section, one may then apply the Holley-Stroock  Lemma to show that the Dirichlet form for the unperturbed measure 
(\ref{fg}) satisfies a log-Sobolev inequality.

\subsection{The convexity-restoring perturbation}

We seek to add a bounded function $W(\phi)$ to $H_{\lambda,\kappa}(\phi)$ such that the sum of 
$\langle \eta, {\rm Hess}_W(\phi)\eta\rangle$ and the quantity in (\ref{focus}) is non-negative. If
$r>  1+ 1/(1-2\gamma)$ and if $\|\phi\|_2^2 > R$, for some sufficiently large $R$ depending on $\lambda$, the quantity in (\ref{focus}) is actually non-negative. We choose such a value of $R$. We are then left with analyzing the Hessian of $H_{\lambda, \kappa}(\phi)$ for
$\|\phi\|_2^2 \leq R$. Here, and only here, do we need help from $W(\phi)$. 

Let $\chi$ be a smooth non-negative cut-off function on $[0,\infty)$  bounded above by $1$, with the properties that  $\chi(t) =1$, for $t\leq 1$, $\chi(t) = 0$, for $t> 2$, and that $|\chi'(t)|,|\chi''(t)| < 5$, for all $t\in [0,\infty)$.   (One may set 
$\chi(t) := 1 - 30\int_1 ^t (1-x)^2(2-x)^2{\rm d}x$, for $1 < t < 2$.)  We then define
$\chi_R(t) = \chi(t/R)$, $R>0$. 

We choose the functional $W(\phi)$ to be given by
\begin{equation}\label{Wdef}
W(\phi) = \frac{c}{2} \left(\sum_{k=-n}^n|\hat\phi(k)|^2\right) \chi_R\left(\sum_{k=-n}^n|\hat\phi(k)|^2\right)\ ,
\end{equation}
where $c$ is a constant to be chosen later. Recall that 
$P_n$ denotes  the orthogonal  projection onto the span of the $\{ e^{-i2\pi kx/L}\ :\ -n \leq k \leq n\}$ in $L^{2}$. \\
By direct calculation,
\begin{eqnarray}\label{modHess}\langle \eta, {\rm Hess}_W(\phi)\eta\rangle   &=&  
c \chi_R\left(\|P_n\phi\|_2^2 \right)\|P_n\eta\|_2^2 \nonumber\\   
   &+ & c  g_1\left(\|P_n\phi\|_2^2\right) \|P_n\eta\|_2^2
+ c g_2\left(\|P_n\phi\|_2^2\right)|\langle \phi,P_n\eta\rangle|^2
 \end{eqnarray}
where $g_1(s) = s\chi_R'(s)$ and $g_2(s) = 2(2\chi_R'(s) + s\chi_R''(s))$.
Note that since $|g_1(s)| + s|g_2(s)| \leq 35$, for all $s$, 
\begin{equation}\label{rembnd}
|c  g_1\left (\|P_n\phi\|_2^2\right) \|P_n\eta\|_2^2
+ c g_2\left(\|P_n\phi\|_2^2\right)|\langle \phi,P_n\eta\rangle|^2| \leq 35 c\|P_n\eta\|_2^2\ .
\end{equation}

The parameters in $W$ are chosen as follows: The parameters $L$,  $m$,  $\kappa$ and $\lambda$  are given. We have already chosen a constant $a>0$ such that the quantity $m_a$, defined in \eqref{madef}, is positive. 
Next, we choose $\gamma\in (1/4,1/2)$ such that  $(1-2\gamma)(r-1) > 1$, and $\epsilon> 0$. This fixes the exponents $q$ and $q'$ in \eqref{cchoice} and \eqref{nchoice}, respectively. As we have noted, these exponents are at least as large as $2$.

\medskip

\noindent{\it (1)} If the quantity in  (\ref{cchoice}) is non-negative, we may choose $c=0$ and $n=\infty$. In this case 
$\lambda$ is so small and $\kappa$ is so large that there is no need to add the functional $W$. Otherwise, we choose $c$ to be minus the quantity in
  (\ref{cchoice}), for the chosen value of  $\gamma$.

\smallskip

\noindent{\it (2)}  Choose $\epsilon = (1-\gamma/2)$, then choose $n$ such that  (\ref{nchoice}) is satisfied for this choice of $\epsilon$. 

\smallskip

\noindent{\it (3)}   Choose $R$ so large that 
$$
 - \left( 3\lambda C_{a,\gamma}^2 L^{4\gamma-1}\right)^{1/(1-2\gamma)} 
R^{2/(1-2\gamma)}
+ \kappa r R^{2r-2} \geq 35c\ .
$$ 
To satisfy this bound when $\lambda$ is not small or when $L$ is large, one needs to choose $r$ such that $r-1 >  1/(1-2\gamma)$, which we have already assumed. 
Since the terms in the second line on the right side of (\ref{modHess}) are bounded by $35 c\|P_n\eta\|_2^2$ and vanish, unless $\|\phi\|_2 > R$, they can be absorbed into positive terms coming from the Hessian of $H_{\lambda, \kappa}$. 

\medskip

With this choice of parameters, we have that
\begin{equation}\label{hesscomp}
{\rm Hess}_{H_{\lambda,\kappa} + W}(\phi) \geq  \alpha {\rm Hess}_{H_0}(\phi) \ .
\end{equation}

\begin{remark}  The size of the constant in the log-Sobolev inequality, and hence the magnitude of the spectral gap will tend to zero exponentially fast in $\|W\|_\infty$. Therefore it is useful to pay attention to how 
$\|W\|_\infty$ depends on the  allowed choices of parameters. First, for given values of $\alpha$ and $L$, there is a constant $\lambda_0(\alpha,L)>0$ such that if $0<\lambda \leq \lambda_0(\alpha,L)$, the quantity in (\ref{cchoice}) is non-negative, and we may set $W = 0$. For large $\lambda$, our prescription yields
$$c = {\mathcal O}\left(\lambda^{\frac{(1-2\gamma)(r-1)}{(1-2\gamma)((1-2\gamma)(r-1)-1)}}\right)\ .$$
In the limit of large $r$, the exponent in this expression approaches $2$, but it is always larger than $2$. We must then choose $R:= {\mathcal O}(c^{1/(2r-2)})$. 
Since $cR/2 \leq \|W\|_\infty \leq cR$, for large $\lambda$, $\|W\|_\infty = \lambda^w$, for some $w>2$, but with $w$ approaching $2$ in the limit $r\to\infty$. 
The log-Sobolev constant and the spectral gap will thus be of order ${\mathcal O}(e^{-K\lambda^w})$, for some constant $K$. 

Finally, we observe that we could have defined $W$ without the projection $P_n$. While it is comforting that, in this problem, we only need help from $W$ for finitely many modes, this is not a necessary condition for the applicability of our strategy.
\end{remark}

\section{Application of the Holley-Stroock Lemma}

Let $(\Omega,\mathcal{F},\mu)$ be a probability space.  We define the functional
$${\rm Ent}_\mu(f)  = \int f\ln f {\rm d}\mu - \left(\int f {\rm d}\mu\right) \ln \left(\int f {\rm d}\mu\right) \ $$
on non-negative functions $f$, with $f(\ln f)_+$ integrable, and we define ${\rm Ent}_\mu(f) $ to be $+\infty$, elsewhere. 
Given a function $f\geq 0$, with $f \ln f$ integrable, we define the function $\varphi$ on $(0,\infty)$ by setting
$$\varphi(t)  = \int \left[ f \ln \left(\frac{f}{t}\right) + t - f\right]{\rm d}\mu\ .$$
Note that $\varphi$ is convex and continuously differentiable, and that
${\displaystyle \varphi'(t) =  -
t^{-1}\int f{\rm d}\mu + 1}$. Hence ${\displaystyle \varphi(t) \geq \varphi\left(\int f{\rm d}\mu\right)}  = {\rm Ent}_\mu(f)$,
for all $t\in (0,1)$.

It follows that, for all non-negative functions $f$ with the property that $f\ln f$ is integrable,
\begin{equation}\label{HS} 
{\rm Ent}_\mu(f)   = \inf_{t\in (0,\infty)} \int_\Omega \left[ f \ln \left(\frac{f}{t}\right) + t - f\right]{\rm d}\mu\ .
\end{equation}

This leads directly to the following lemma; in our applications, the quadratic fuction $\Gamma(f,f)$ in the lemma will be $\langle D F,\sigma^2 D F\rangle$. 

\begin{lm}[Holley-Stroock Lemma]  Let $(\Omega,{\mathcal F},\mu)$ be a probability space on which there is a densire subset
$\mathcal{D}$ of $L^2((\Omega,{\mathcal F},\mu)$ on which there is defined a real bilinear map 
$f\mapsto \Gamma(f,f)\in L^1(\Omega,{\mathcal F},\mu)$. Suppose further that $F\mapsto \int_\Omega \Gamma(f,f){\rm d}\mu$
is a Dirichlet from on $L^2((\Omega,{\mathcal F},\mu)$, and that  that the log-Sobolev inequality
$${\rm Ent}_\mu(f^2)   \leq c \int_\Omega \Gamma(f,f) {\rm d}\mu$$
is valid.  Let $V$ be a continuous function with finite oscillation,
$${\rm osc}(V) := \sup V - \inf V\ ,$$
and define a new probability measure $\widetilde \mu$ by
$\widetilde \mu = \frac1Z e^V \mu$.
Then the logarithmic Sobolev inequality
$${\rm Ent}_{\widetilde \mu}(f^2)   \leq c e^{{\rm osc}(V)} \int_\Omega\Gamma(f,f){\rm d}\widetilde \mu$$
is valid
\end{lm}

\begin{proof}   Note that $f^2(\ln f^2)_+$ is integrable with respect to $\mu$ if and only if it is integrable with respect to 
$\widetilde \mu$, so that ${\rm Ent}_{\widetilde \mu}$ and ${\rm Ent}_{\mu}$ have the same domain of definition. 
By (\ref{HS}), and since  the integrand is non-negative,
\begin{align*}{\rm Ent}_{\widetilde \mu}(f^2)   &= 
\inf_{t\in (0,\infty)} \int_\Omega \left[ f^{2} \ln \left(\frac{f^{2}}{t}\right) + t - f^{2}\right]\frac1Z e^{V}{\rm d}\mu\\
&\leq \frac1Z e^{\sup V}  \inf_{t\in (0,\infty)} \int_\Omega \left[ f^{2} \ln \left(\frac{f^{2}}{t}\right) + t - f^{2}\right]{\rm d}\mu =
 \frac1Z e^{\sup V} {\rm Ent}_{ \mu}(f^2)\ .
\end{align*}
Even more simply,
$$  \frac1Z e^{\sup V} \int \Gamma(f,f){\rm d}\mu  = e^{\sup V}\int_\Omega \Gamma(f,f) e^{-V} 
{\rm d}\widetilde\mu    \leq e^{{\rm osc} V} 
\int_\Omega \Gamma(f,f){\rm d}\widetilde \mu\ .$$
Combining these bounds completes the proof of the lemma.
\end{proof}

We apply this lemma with $\mu := \text{d}\mu_{\lambda,\kappa}$, as introduced in Eq. \eqref{fg}, and $\Gamma(f,f) := \langle D F,\sigma^2 D F\rangle$. Recall that
$$H_{\lambda,\kappa}(\phi)  =  H_0(\phi) + \kappa \|\phi\|_2^{2r}-   \frac{\lambda}{p} \|\phi\|_p^p\ .$$
To $H_{\lambda,\kappa}(\phi)$ we add the functional
\begin{equation}\label{Wdef}
W(\phi) = a (\sum_{k=-n}^n|\hat\phi(k)|^2) \chi_R(\sum_{k=-n}^n|\hat\phi(k)|^2)\ .
\end{equation}
Let $P_n$ be the projector onto the span of the $\{ e^{-i2\pi kx/L}\ :\ -n \leq k \leq n\}$. Then

$${\rm Hess}_{W}(\phi)  =  2a \chi_R(\sum_{k=-n}^n|\hat\phi(k)|^2) P_n + 
2a g_1(\sum_{k=-n}^n|\hat\phi(k)|^2) P_n  + 
4a g_2(\sum_{k=-n}^n|\hat\phi(k)|^2) | P_n\phi \rangle\langle P_n\phi|\ ,$$ 
where $g_1(s) = \chi_R'(s)$ and $g_2(s) = 2\chi_R'(s) + s\chi_R''(s)$.

To estimate the Hessian of $H_{\lambda,\kappa}(\phi)  + W(\phi)$, we return to \eqref{bound3} and 
make two changes: First, we add the additional terms
due to the inclusion of $W$. Second, we use the spectral decomposition to estimate
the term
$$
\lim_{t\to0}\frac{1}{t^2}3\lambda \left( \frac12[ \|\phi+t\eta\|_4^4 + \|\phi-t\eta\|_4^4] - \|\phi\|_4^4\right)
= 3\lambda \int |\phi|^2|\eta|^2\ .
$$
We  use Lemma~\ref{split} to show that
$$3\lambda \int |\phi|^2|\eta|^2  \leq  
 C_\gamma  \|P_n\psi\|_2^{2-4\gamma}\|P_n\psi'\|_2^{4\gamma}\|\phi\|_2^2 
+ C_\gamma \|P_n^\perp\psi\|_2^{2-4\gamma}\|P_n^\perp \psi'\|_2^{4\gamma}\|\phi\|_2^2\ .
$$
We require positivity of $S_1+ S_2$, where
\begin{eqnarray}\label{S1}
S_1(\eta) &:= &(1-\alpha)\frac12 \|P_n^\perp \eta'\|_2^2 + (1-\alpha)\frac{m}{2} \|\eta\|_2^2  - \nonumber \\
        &  & 3\lambda C_\gamma^2 \|P_n^\perp\eta\|_2^{2-4\gamma}\|P_n^\perp\eta'\|_2^{4\gamma} \|\phi\|_2^2 + \kappa r \|\phi\|_2^{2r-2}  \|P_n^\perp\eta\|_2^2 \ ,
\end{eqnarray}
and
\begin{multline}\label{bound4B}
S_2(\eta) := (1-\alpha)\frac12 \|P_n\eta'\|_2^2 + (1-\alpha)\frac{m}{2} \|P_n\eta\|_2^2  - 3\lambda C_\gamma^2 \|P_n\eta\|_2^{2-4\gamma}\|P_n\eta'\|_2^{4\gamma} \|\phi\|_2^2 + \kappa r \|\phi\|_2^{2r-2}  \|\eta\|_2^2 +\\
\hspace{0.2cm}
2a \chi_R(\sum_{k=-n}^n|\hat\phi(k)|^2) \|P_n\eta\|_2^2  + 2a g_1(\sum_{k=-n}^n|\hat\phi(k)|^2) \|P_n\eta\|^2
+ 4a g_2(\sum_{k=-n}^n|\hat\phi(k)|^2) | \langle P_n\phi,\eta\rangle|^2\ .
\end{multline}
It suffices to show that, for some $n$  and appropriate choices of the other parameters, $S_1$ and $S_2$ are positive. 

First, we consider $S_1$. Since
${\displaystyle \|P_n^\perp\eta'\|_2^2 \geq \frac{1}{(2\pi n/L)^2} \|P_n^\perp\eta\|_2^2}$, 
\begin{equation*}
\|P_n^\perp\eta'\|_2^{4\gamma} \leq \frac{1}{(2\pi n/L)^{4\epsilon}}
\|P_n^\perp\eta\|_2^{-4\epsilon}\|P_n^\perp\eta'\|_2^{4(\gamma+\epsilon)},
\end{equation*}
hence $S_1 \geq S_1{'}$, where 
\begin{multline}\label{S1}
S_1{'} := (1-\alpha)\frac12 \|P_n^\perp \eta'\|_2^2 + (1-\alpha)\frac{m}{2} \|\eta\|_2^2  -\\ 3\lambda C_\gamma^2 
\frac{1}{(2\pi n/L)^{4\epsilon}} \|P_n^\perp\eta\|_2^{2-4(\gamma-\epsilon)}\|P_n^\perp\eta'\|_2^{4(\gamma+\epsilon)} \|\phi\|_2^2 + \kappa r \|\phi\|_2^{2r-2}  \|P_n^\perp\eta\|_2^2 \ .
\end{multline}
Choosing $n$ sufficiently large, we can effectively make $\lambda$ arbitrarily small, and then positivity of $S_1{'}$ follows from our previous result. Turning to $S_2$, we observe that the inclusion of $W$ effectively makes the mass in $S_2$ arbitrarily large, and hence, once again, our previous analysis establishes the positivity of $S_2$.  Altogether, this completes the proof of the main theorem. 

\vspace{0.2cm}
\textit{Acknowledgements:} EA was partially supported by NSF grant DMS 1501007.   JF thanks the Institute for Advanced Study, Princeton, and, in particular, Thomas C. Spencer for splendid hospitality during a period when the work underlying this paper was begun. The work of JLL was supported in part by AFOSR grant FA9550-16-10037, and was carried out, in part, while he was visiting the Institute for Advanced Study.

\section {Appendix: Spectral Gap and Witten Laplacian} 

In this section, we briefly recapitulate a formulation of the problem of exhibiting a gap above the ground-state energy of our Hamiltonian in terms of the Witten Laplacian. The material reviewed here and in Section 5.1 is standard
and is similar to the contents of Section 7 in [14]. We add it here to fix our notations and for the convenience of the reader.\\ 
We start our review by considering systems with only finitely many degrees of freedom. 
It will turn out to be convenient to re-write our Hamiltonian in Fourier modes. For ease of exposition, we consider the cubic NLS, with $p=4$, and we set $\beta=m=1$.
The resulting Hamiltonian, denoted by $2\Phi$, is then given by
\begin{eqnarray} 2\Phi(a,\bar a)\label{ham2}
&=&\sum_{n\in\Bbb Z}(n^2+1)|a_n|^2-\frac{\lambda}{2}\sum_{n_1-n_2+n_{3}-n_{4}=0} a_{n_1}\bar a_{n_2} a_{n_{3}}\bar a_{n_{4}}  \nonumber\\
& &+\frac {\kappa}{r+1}(\sum_na_n\bar a_n)^{r+1}\,\end{eqnarray}
where $a_n, \bar a_n\in\Bbb C$.  
This definition differs from the one in (\ref{ham}) by a factor of $2$ and $r$ is replaced by $r+1$, which slightly simplifies some of the factors later on, in sect. 5.2.

Consider a truncated Hamiltonian, $\Phi_{N}$, instead of
$\Phi$, which we define to be given by
$$\Phi_N=\Phi|a_n=0, \bar a_n=0,   \quad \text{  for   } |n|>N.$$  
For simplicity of notation, we drop the subscript $N$ below. 
Up to a normalization factor, the truncated Gibbs measure takes the form 
\begin{equation}
\label {mu}
\mu \propto e^{-2\Phi}\prod_n da_n d\bar a_n.
\end{equation}
When identifing $\mathbb C^{2N+1}$ with $\mathbb R^{2(2N+1)}$, $\mu$ is a probability measure on 
$\mathbb R^{2(2N+1)}$.
We will show that, up to a normalization constant, $e^{-\Phi}$ is the {\it unique} ground-state 
of a certain Schr\"odinger operator, which is, in fact, the generator of a diffusion process, (cf. $\mathcal L$ in sect. 1).
To provide precise ideas, we need to engage on a short digression and introduce some notions and notations.

\subsection {Some elements of differential calculus on $\mathbb{R}^{N}$}

In this section, we review some basic elements of differential calculus on $\mathbb{R}^{N}$. We equip $\mathbb{R}^{N}$ with the standard euclidian metric, $(\delta_{ij})_{i,j=1}^{N}$. Let $\phi$ be a smooth real-valued function on $\mathbb{R}^{N}$, i.e., \mbox{$\phi\in C^\infty(\Bbb R^{N}; \Bbb R)$.} Let $d$ be the usual exterior differentiation
$$d=\sum^N_{j=1}dx^j{\wedge}\partial_{x_{j}}(\cdot), $$
 and 
$$d_{\phi}=e^{-\phi}  d e^{\phi}=d+d\phi{\wedge}=\sum_{j=1}^Ndx^j{\wedge}z_j(\cdot),$$ 
where
$$z_j=\frac {\partial}{\partial x^j}+\frac{\partial\phi}{\partial x^j}.$$ 
For details concerning differential calculus, see for example \cite{Sp}.

If $f$ is a form of degree $m$, then $d_\phi f$ is a form of
degree $m+1$. For example, if $f$ is a $0$-form, i.e., a scalar function in $C^\infty(\Bbb R^N; \Bbb R)$, 
then
$$d_\phi f=\sum_{j=1}^N z_j (f) dx^j$$
is a $1$-form, which we may identify with a covariant vector-valued function, $F$, with components
$$F_j(x)=z_j (f) (x),$$
which are functions in $C^\infty(\Bbb R^{N}; \Bbb R^N)$.
We note that if $\phi=0$ then $d_\phi f=df$, which is just the usual differential of $f$.
If $f$ is a $1$-form, $f=\sum_j f_j dx^j$, then 
$$d_\phi f=\sum_{i<j} z_i(f_j) dx^i\wedge dx^j$$
is a $2$-form, which we may identify with an $N\times N$ antisymmetric matrix function, $M$, with matrix elements
$$M_{ij}(x)=-M_{ji}(x)=z_i(f_j) (x),$$ i.e., $M$ is a function in $C^\infty(\Bbb R^{N}; \Bbb R^N\wedge\Bbb R^N)$.
In view of its action on $e^{-\phi}$, the operator $z_j$ can be interpreted as an ``annihilation operator'': 
$$z_j e^{-\phi}=0, \quad \text{ for } j=1, 2, ...,N. $$ 

The space of $m$-forms, $m=1,\dots, N$, can be equipped with an $L^2$- scalar product:
For two $m$-forms, $\omega$ and $\nu$, the scalar product, $(\omega, \nu)$, is defined by
\begin{equation}\label{scalarprod}
(\omega, \nu):= \int \omega \wedge {*}\nu,
\end{equation}
where ${*}$ is the Hodge *-operation, (which involves the metric $(\delta_{ij})$ on $\mathbb{R}^{N}$).
Choosing $\nu=d_\phi f$,
with $f$ an $(m-1)$-form, we may introduce the adjoint, $d^{*}_{\phi}$, of the operator $d_{\phi}$ by setting

$$(d^{*}_{\phi}\omega, f):=(\omega, d_{\phi}f).$$
Thus,
$$d^{*}_{\phi}=e^{\phi} 
d^{*}e^{-\phi}=\sum^N_{j=1}dx^{j}\,\rfloor \,z^{*}_j(\cdot),$$ 
where 
$$z^{*}_j=-\frac {\partial}{\partial x^j}+\frac {\partial\phi}{\partial x^j},$$ 
(recall that the metric is given by $(\delta_{ij})$), and ``$\rfloor$'' is the usual interior multiplication, which lowers the degree of forms by one.\\
If $\omega$ is a form of degree $m$, then $d^*_\phi \omega$ is a form of
degree $m-1$. For example, if $$\omega=\sum^N_{j=1} \omega_jdx^j$$ 
is a $1$-form, then 
$$d^*_\phi \omega=\sum_{j=1}^N z^*_j (\omega_j)$$ 
is a $0$-form, i.e., a scalar function in $C^\infty(\Bbb R^{N}; \Bbb R)$.
(If $\omega$ is a 1-form then, for $\phi:=0$, $d^*_\phi \omega=d^* \omega$ is just the
``divergence'' of $\omega$.) If $\omega$ is a $0$-form, then $d^*_\phi \omega=0$.

The operator $z^{*}_j$ can be interpreted as a ``creation operator''.
For example, if $N=1$ and $\phi=x^2$, then $z^{*}:=z^{*}_j$ generates the
first Hermite polynomial. The operators $z_j, z_{j}^{*} , \, j=1,\dots, N,$ satisfy the canonical commutation relations: 
\begin{equation}\label{z}
 [z_j,z_k^{*}]=2\partial_j\partial_k\phi.
\end{equation}
One easily checks that the operators $d_{\phi}$ and $d_{\phi}^{*}$ are nilpotent, i.e.,
 $$d_\phi d_\phi=d^{*}_{\phi}d^{*}_{\phi}=0.$$ 
The space of smooth differential forms is defined by
$$\Omega(\mathbb{R}^N):= \bigoplus_{\ell =1}^{N} \mathcal{S}(\mathbb{R}^{N};(\mathbb{R}^{N})^{\wedge \ell}),\quad \text{ where  }\,\,(\mathbb{R}^{N})^{\wedge \ell}:=
\underbrace{\mathbb{R}^{N}\wedge \cdots \wedge \mathbb{R}^{N}}_{\ell \,\text{ times}}.$$
Here $\mathcal{S}$ denotes Schwartz space. On the space $\Omega(\mathbb{R}^{N})$ of differential forms we define the ``Witten Laplacian''
\begin{equation}\label{w}
\Delta_{\phi}=d^{*}_{\phi}{d_\phi}+{d_\phi}d^{*}_{\phi}.
\end{equation}
Notice that 
\begin{equation}\label{d} 
d_\phi \Delta_{\phi}=\Delta_{\phi}d_\phi \quad\text{and }\quad d_\phi^{\ast} 
\Delta_{\phi}=\Delta_{\phi}d_\phi^{*},
\end{equation}  
where one uses (\ref{w}). More precisely, denoting by
$\Delta_{\phi}^{(\ell)}$ the restriction of the Witten Laplacian $\Delta_{\phi}$ to forms of degree
$\ell$, we have that
$$d_\phi \Delta_{\phi}^{(\ell)}=\Delta_{\phi}^{(\ell+1)}d_\phi,\,d_\phi^{*} 
\Delta_{\phi}^{(\ell+1)}=\Delta_{\phi}^{(\ell)}d_\phi^{*}.$$ 
The standard Hodge Laplacian corresponds to setting $\phi=0$.
For a quick overview of analytical aspects of Hodge theory, see Chapt. 11.3 in \cite{CFKS}.\\
The explicit expression for $\Delta_{\phi}^{(0)}$ is given by
$$\Delta_\phi^{(0)}={d^*_\phi}{d_\phi}=\sum^N_{j=1} 
z_j^{*}z_j=-\sum^N_{j=1}\frac{\partial^2}{\partial x_j^2}+\Vert 
d\phi\Vert^2-\text {Tr} \text{ Hess }\phi.$$ For example, if $\phi$ is a
non-degenerate quadratic function on $\mathbb{R}^{N}$, then
$\Delta_\phi^{(0)}$ is the Hamiltonian of $N$ harmonic oscillators, and $z_j$  and
$z_j^{*}$ are the usual annihilation/lowering- and creation/raising operators of $N$ harmonic oscillators, respectively.\\
More generally, we have that
$$\aligned \Delta_\phi&=\sum\sum z_j 
z_k^{*}dx_j^{\wedge}dx_{k}^{\rfloor}+\sum\sum 
z_k^{*}z_jdx_{k}^{\rfloor}dx_j^{\wedge}\\ &=\sum\sum 
z_k^{*}z_j(dx_j^{\wedge}dx_{k}^{\rfloor}+dx_{k}^{\rfloor}dx_j^{\wedge})+
[z_j,z_k^{*}]dx_j^{\wedge}dx_{k}^{\rfloor}\\
&=\sum z_j^{\ast}z_j+2\sum\sum(\partial_{x_j}\partial_{x_k}\phi) 
dx_j^{\wedge}dx_{k}^{\rfloor}\\ &=\Delta_\phi^{(0)}\otimes
\Bbb I+2\sum\sum(\partial_{x_j}\partial_{x_k}\phi) 
dx_j^{\wedge}dx_{k}^{\rfloor},\endaligned$$ 
where, to obtain the third line from the second line, we have used (\ref{z}). In particular, with the identification of 1-forms with covariant-vector-valued functions on $\Bbb R^{N}$, we find that
\begin{equation}\label{w1}
\Delta_{\phi}^{(1)}=\Delta_{\phi}^{(0)}\otimes \Bbb I+2\,\text{Hess  }\phi.
\end{equation} 

For a smooth, polynomially bounded function $\phi$, $(\Delta_{\phi}^{(\ell)}\omega,\omega)\geq 0,$ for an arbitrary $\ell$-form $\omega \in \Omega(\mathbb{R}^{N})$, and 
$\Delta_{\phi}^{(\ell)}$ is a non-negative, self-adjoint operator on a dense domain in the Hilbert-space completion of the space $\Omega(\mathbb{R}^{N})$ with respect to the scalar product introduced in \eqref{scalarprod}. If the function $\phi$ grows like a positive (fractional) power of $\vert x \vert$ then the operators $\Delta_{\phi}^{(\ell)}$ have compact resolvents and hence their spectra are discrete and contained in $[0, \infty)$; cf. \cite{Sj}. 
The lowest eigenvalue of $\Delta_{\phi}^{(0)}$ is zero, and the corresponding eigenstate is given by
$Z e^{-\phi}$, where $Z$ is a normalization factor. This state is annihilated by $d_{\phi}$. The  eigenvalue $0$ is
simple; for, if $u$ is another eigenfunction corresponding to the eigenvalue $0$, then
$0=(\Delta_{\phi}^{(0)}u,u)=\Vert  d_{\phi} u\Vert^{2}$ and hence $d_{\phi} u=0$,
which implies that $u$ is a multiple of $e^{-\phi}$.

Using (\ref{d}), we obtain the following intertwining property of the spectra:
\begin{equation}\label{i}\sigma (\Delta_\phi^{(0)})\backslash\{0\}\subset \sigma (\Delta_\phi^{(1)}).\end{equation}
This is because if $u$ is an eigenfunction of $\Delta_\phi^{(0)}$, i.e.,
$$\Delta_\phi^{(0)}u=\kappa u$$
corresponding to an eigenvalue $\kappa>0$ then, applying $d_\phi$ to both sides, we find that
$$d_\phi \Delta_\phi^{(0)}u=(d_\phi d^*_\phi) d_\phi u=\Delta_\phi^{(1)} (d_\phi u)=\kappa d_\phi u.$$
Thus, if $\kappa\neq 0$ then $d_\phi u\neq 0$ is an eigenform for  $\Delta_\phi^{(1)} $, which is the 
statement in (\ref{i}). 
Using (\ref{w1}), we conclude that $\sigma (\Delta_\phi^{(0)})$ has a spectral gap if 
$\phi$ is strictly convex. (This implication is the main reason why we have introduced $\Delta_\phi^{(1)}$.)  

Replacing $N$ by $2(2N+1)$ and setting $\phi:=\Phi$, we observe that
$e^{-\Phi}$ is (proportional to) the ground-state eigenfunction of the Witten Laplacian 
$\Delta_{\Phi}^{(0)}$, which is a Schr\"odinger operator with potential 
$$V= \Vert d\Phi\Vert^2-\text {Tr} \text{ Hess }\Phi.$$

We note that the operator $\Delta_{\Phi}^{(0)}$ coincides with a truncation of the operator
$\mathcal L$ introduced in sect. 1, provided $\sigma$ is chosen to be the identity operator; (cf. sect. 2 of \cite{LMW}). 

\subsection {A quantitative estimate on the spectral gap} 

We now apply the formalism introduced above to estimate the spectral gap of the generator $\mathcal{L}$, see \eqref{noise3}, of the stochastic process introduced in  (\ref{noise2}). When expressed in terms of Fourier modes, the metric $\sigma^2$ (see \eqref{metric}) is a constant diagonal matrix given by
$$\widehat {\sigma^2}(n, n):=\sigma^2(n, n)=(n^2+1)^{-s}, s>0.$$
Let $d$ denote exterior differentiation, as above. In terms of Fourier modes, it is given by 
$$d=\sum_ndb_{n}\wedge\partial_{b_n}(\cdot),$$
where $b_n$ stands for either $a_n$ or $\bar a_n$, and  
$$d_{\Phi}:=e^{-\Phi}d e^{\Phi}=\sum_n db_{n} \wedge (\partial_{b_n}+\partial_{b_n}\Phi )(\cdot).$$ 
We introduce the ``metric'' 
$$A:=\begin{pmatrix}  \sigma^2&0\\0&\sigma^2\end{pmatrix},$$
where each block corresponds to one 
of the four possible ``sectors'' $\bar a a, \bar a\bar a, aa, a\bar a$;
(we recall the identification of $\mathbb C^{2N+1}$ with $\mathbb R^{2(2N+1)}$ introduced earlier).
We define
the (formal) adjoint of $d_\Phi$ with respect to $A$ to be:
$$d_\Phi^{*}= \sum_n(-\partial_{\bar b_n}+\partial_{\bar b_n}\Phi )\circ A d\bar b_n^{\rfloor}(\cdot)\,.$$ 
The Witten-Laplacian is defined by 
$$\Delta_{\Phi}=d_\Phi^{*}d_\Phi+d_\Phi d_\Phi^{*}.$$
Rather straightforward computations show that the restrictions of the Witten Laplacian to the spaces of $0$-forms and $1$-forms are given by
$$\Delta^{(0)}_\Phi=-2\sum_n\frac{\partial^2}{\partial \bar a_n\partial a_n}+\langle Ad\Phi, d\Phi\rangle - 
\text{Tr }(\text{Hess  }\Phi\circ A),$$
\begin{equation}\label{l1}\Delta^{(1)}_\Phi=\Delta^{(0)}_\Phi\otimes\Bbb I+2\,\text{Hess   }\Phi \circ A,\end{equation}
where the Hessian matrix is given by
\begin{equation}\label{phi}
\text{Hess   }\Phi=M_1+M_2,
\end{equation} 
with
$$M_1=\begin{pmatrix} [[\partial_{\bar a_j}\partial _{ a_k}\Phi]]&0\\0&[[\partial_{a_j}\partial_ {\bar a_k}\Phi]]\end{pmatrix},$$
$$M_2= \begin{pmatrix} 0&[[\partial_{\bar a_j}\partial_{\bar a_k}\Phi]] \\   [[\partial_ {a_j}\partial_{a_k}\Phi]]&0\end{pmatrix},$$
and $[[\quad]]$ denotes a matrix of second order partial derivatives.
Notice that $\Delta^{(0)}_\Phi$ coincides with the operator $\mathcal L$ introduced in sect. 1 and that spectral gap
above the ground-state energy of $\mathcal{L}$ governs the exponential rate of approach to equilibrium. In the following proposition, the constants $\lambda$, and $\kappa$ are as in \eqref{ham}; (while $r$ differs by $1$, cf. (\ref{ham2})). 

\begin {prop}\label{GAP} 
Up to constant multiples, the function $e^{-\Phi}$ is the unique eigenfunction of $\Delta_\Phi^{(0)}$ corresponding to the eigenvalue $0$. The smallest strictly positive
eigenvalue, $E_1$, of $\Delta_\Phi^{(0)}$ satisfies the lower bound
$$E_1\geq 1-\big(\frac{\lambda} {\epsilon}\big)^{\frac{r}{r-1}}\big(\frac{r-1}{r}\big)\frac{1}{(\kappa r)^{\frac{1}{r-1}}}>0,$$
provided $0<\epsilon<1$, $\lambda$ is chosen small enough, $r\geq \frac{ 2}{1-\epsilon}$, and $s\leq 1$.
\end{prop}

\begin {remark} 
Note that $E_1\geq 1-\lambda/\epsilon$, as $r\to\infty$. In this limit, $E_1$ ought to correspond to
the smallest non-zero eigenvalue of the operator $\Delta_\Phi^{(0)}$, with $\phi$ (see \eqref{ham}) restricted to a ball of radius $1$ and Dirichlet boundary conditions imposed on the Laplacian acting on $\phi$.
\end{remark}

\begin{proof} 
The statements that the eigenvalue $0$ is simple and that the spectra of $\Delta_{\Phi}^{(0)}$ and 
$\Delta_{\phi}^{(1)}$ are related by
 $$\sigma (\Delta_\phi^{(0)})\backslash\{0\}\subset \sigma (\Delta_\phi^{(1)})$$
are proven as explained above; (our arguments are independent of the choice of the metric $A$). 

Using (\ref{l1}), one observes that if there exists a constant $c>0$ such that, for all $w$, 
$$\langle A\bar w, 2\,\text{Hess   }\Phi\circ Aw\rangle \geq c\langle A\bar w,w\rangle,$$
then $E_1\geq c>0$.
To apply this abstract argument to our concrete example, we need to make some explicit computations using (\ref{phi}). We write 
$$A\circ 2\,\text{Hess  }\Phi\circ A:=\mathcal M=\begin{pmatrix} \mathcal M_{11}&\mathcal M_{12}\\\mathcal M_{21}&\mathcal M_{22}\end{pmatrix}, \qquad w=\begin{pmatrix} u\\\bar{u} \end{pmatrix}.$$
Then
$$\langle \bar w, \mathcal M w\rangle=2 \langle \bar u, \mathcal M_{11} u\rangle+ 2\text{Re }\langle \bar u, \mathcal M_{12} \bar u\rangle.$$
The matrix elements of $\mathcal{M}$ can be seen to be as follows:  
$$\mathcal M_{11}(n, m):=D_n\delta_{nm}-B_{nm}+C_{nm},$$
with 
$$\aligned &D_n=(n^2+1)^{1-2s}+(n^2+1)^{-2s}\kappa (\sum a_\ell\bar a_\ell)^r\\
&B_{nm}=(\lambda/2)(n^2+1)^{-s}(m^2+1)^{-s}\sum_{k-\ell=n-m} a_\ell\bar a_k\\
&C_{nm}=(n^2+1)^{-s}(m^2+1)^{-s}\kappa r(\sum a_\ell\bar a_\ell)^{r-1}a_n\bar a_m,\endaligned$$
and 
$$\mathcal M_{12}(n, m):=-B'_{nm}+C'_{nm},$$
where 
$$\aligned 
&{B'}_{nm}=(\lambda/2)(n^2+1)^{-s}(m^2+1)^{-s}\sum_{k+\ell=n+m} a_\ell a_k\\
&{C'}_{nm}=(n^2+1)^{-s}(m^2+1)^{-s}\kappa r(\sum a_\ell\bar a_\ell)^{r-1}a_n a_m\,.\endaligned$$
Since 
$$\langle \bar u, Cu\rangle+ \text{Re } \langle \bar u, C'\bar u\rangle\geq 0,$$
we have that
$$\aligned & \langle \bar u, \mathcal M_{11} u\rangle+ \text{Re }\langle \bar u, \mathcal M_{12} \bar u\rangle\\
& \geq \langle \bar u, (D-B) u\rangle - \text{Re }\langle \bar u, {B'} \bar u\rangle\\
& :=I.\endaligned$$
Let $\tilde u$ be the function with Fourier coefficients 
$$\hat{\tilde u}_n=(n^2+1)^{-s}u_n.$$ 
Then 
$$\langle \bar u, D u\rangle=\Vert \tilde u\Vert ^2_{H_1}+\kappa \Vert a \Vert _{2}^{2r}\Vert \tilde u\Vert ^2_{2},$$
and, similarly, 
$$\aligned &|\langle \bar u, B u\rangle|+| \text{Re } \langle \bar u, {B'} \bar u\rangle|\\
& \leq \lambda\Vert \tilde u a\Vert_2^2\\
& \leq \lambda\Vert \tilde u \Vert _\infty^2 \Vert a\Vert_2^2\\
& \leq \frac{\lambda}{\epsilon}\Vert \tilde u \Vert _{H^{\frac{1+\epsilon}{2} } }^2 \Vert a\Vert_2^2, \endaligned$$
for $\epsilon > 0$. Thus
$$\aligned I&\geq \Vert \tilde u\Vert ^2_{H_1}-\frac{\lambda}{\epsilon}\Vert \tilde u \Vert _{H^{\frac{1+\epsilon}{2} } }^2 \Vert a\Vert_2^2+\kappa  \Vert a \Vert _{2}^{2r}\Vert \tilde u\Vert ^2_{2}\\
&\geq \Vert \tilde u\Vert ^2_{H_1}-\frac{\lambda}{\epsilon}\Vert \tilde u \Vert _{H^1}^{1+\epsilon}  \Vert \tilde u \Vert _{2}^{1-\epsilon}\Vert a\Vert_2^2+\kappa \Vert a \Vert _{2}^{2r}\Vert \tilde u\Vert ^2_{2}, \endaligned
$$
with $\epsilon>0$. We set
$$\Vert \tilde a\Vert ^2_{2}=:K\frac{\Vert \tilde u\Vert ^{1-\epsilon}_{H^1}} {\Vert \tilde u\Vert ^{1-\epsilon}_{2}}>0.$$
$$\aligned \text{RHS}=&(1-\frac{\lambda}{\epsilon}K)\Vert \tilde u \Vert^2 _{H^1}\\
&+\kappa K^r  \frac{\Vert \tilde u\Vert ^{r(1-\epsilon)-2}_{H^1}} {\Vert \tilde u\Vert ^{r(1-\epsilon)-2}_{2}}    \Vert \tilde u \Vert^2 _{H^1}\\
&\geq (1-\frac{\lambda}{\epsilon}K+\kappa K^r  )\Vert \tilde u \Vert^2 _{H^1}, \endaligned$$
if $r\geq \frac{2}{1-\epsilon}$, $0<\epsilon<1$. Here we have used that 
$$\frac{\Vert \tilde u \Vert_{H^1}} {\Vert \tilde u \Vert _{2}}\geq 1.$$

Let $$X=1-\frac{\lambda}{\epsilon}K+\kappa K^r, \quad K>0.$$
Setting $\partial X/\partial K=0$, leads to 
$$K=\big[\frac{\lambda} {\kappa r \epsilon}]^{\frac{1}{r-1}}.$$
Since $$\partial ^2 X/\partial K^2>0,$$ this yields
$$X_{\text {min}}=1-\big(\frac{\lambda} {\epsilon}\big)^{\frac{r}{r-1}}\big(\frac{r-1}{r}\big)\frac{1}{(\kappa r)^{\frac{1}{r-1}}}:=c>0,$$
for $\lambda$ small enough and $r\geq \frac{2}{1-\epsilon}$, with $0<\epsilon<1$. 
We conclude that
$$\langle A\bar w, 2\,\text{Hess  }\Phi\circ A w\rangle\geq c\Vert \tilde w\Vert _{H_1}^2\geq c\langle \bar w, A w\rangle,$$
where $\Vert \tilde w\Vert _{H_1}^2=\Vert \tilde u\Vert _{H_1}^2+\Vert \tilde {\bar u}\Vert _{H_1}^2$,
provided $1-2s\geq -s$, or $s\leq 1$.
\end{proof}

Let $f_0\mu$ be the distribution of the initial data $u_0$ for the  stochastic NLS in (\ref{noise2}), where $\mu$ is the normalized Gibbs measure in (\ref{mu}). Let 
$f_t\mu$ be the distribution of the solution at time $t$, $u_t$. One then has the following result on exponential convergence to the Gibbs state.

\begin{cl}
$$\Vert f_t-1\Vert_{L^2(\mu)}\leq e^{-tE_1}\Vert f_0-1\Vert_{L^2(\mu)}, $$
where $E_1>0$ satisfies the lower bound in Proposition \ref{GAP}, provided $0<\epsilon<7/9$ and $\lambda$ is sufficiently small. 
\end{cl}

\begin{proof} This follows from Theorem 5.7 in \cite{CFL}, which, thanks to Proposition \ref{GAP}, can be applied 
provided $r>9$, since the lower bound on $E_1$ is then {\it uniform} in the truncation of the Fourier modes at $n=N$.  
\end{proof}



\vspace{0.5cm}
$^{1}$ Department of Mathematics, Rutgers University, Hill Center - Busch Campus, Piscataway, NJ 08854-8019, USA; Email: carlen@math.rutgers.edu\\
$^{2}$ Institute for Theoretical Physics, HIT K42.3, ETH Zurich, CH-8093 Zurich; Email: juerg@phys.ethz.ch\\
$^{3}$ Departments of Mathematics and Physics, Rutgers University, Hill Center - Busch Campus, Piscataway, NJ 08854-8019, USA; Email: lebowitz@math.rutgers.edu\\
$^{4}$ CNRS and Department of Mathematics, Universit\'e Cergy-Pontoise, 95302 Cergy-Pontoise Cedex, France; Email: Wei-min.Wang@math.u-psud.fr\\

\begin{thebibliography}{99}

\bibitem{Bl} G. Blower, {\it Logarithmic Sobolev inequality for the invariant measure of the periodic Korteweg--de Vries equation.}, Stochastics  {\bf 84} , p. 533-542 (2012)

\bibitem{BBD}  G. Blower, C. Brett, C.and I. Doust, {\it Logarithmic Sobolev inequalities and spectral concentration for the cubic Shršdinger equation}
Stochastics. {\bf 86},  p. 870-881 (2014)

 
 
\bibitem{B} C. Borell, {\it Convex set functions in $d$-space}, Period. Math. Hungar. {\bf 6}, no. 2, (1975), 111-136.
 
 
\bibitem{C} L. Caffarelli, {\it  Monotonicity Properties of Optimal Transportation¦and the FKG and Related Inequalities}, Comm. Math. Phys. 
{\bf 214}, no. 3 (2000) 547-563.


\bibitem{CFKS} H. L. Cycon, R. G. Froese, W. Kirsch, B. Simon, {\it  Schr\"odinger Operators}, Springer-Verlag, 1987. 


\bibitem{CFL} E.A. Carlen,  J. Fr\"ohlich and   J.L. Lebowitz, {\it Exponential Relaxation to Equilibrium for a One-Dimensional Focusing Non-Linear Schršdinger Equation with Noise},
Comm. Math. Phys. {\bf 342}, no. 1, (2016) 303-332.

\bibitem{F69} Federbush, I.; A partially alternative derivation of a result of Nelson, J. Math. Phys {\bf 10}, 50-52 (1969) 



\bibitem{FU1}D. Feyel and A.S.  \"Ust\"unel,  {\it The notion of convexity and concavity on Wiener space}. Jour. Func. Analysis,  {\bf 176}, (200), 400-428.
 
  \bibitem{FU2}D. Feyel and A.S.  \"Ust\"unel,  {\it Monge-Kantorovitch measure transportation and Monge-Amp\'ere equation on
Wiener space}. Prob. Theor. Rel. Fields,  {\bf 128}, no. 3, (2004), 347-385.

\bibitem{HS} R.~Holley and D.~Stroock,  . Logarithmic Sobolev inequalities and stochastic Ising models. J. Stat. Phys. {\bf 46} (1987) 1159-1194.

 \bibitem{FU3}D. Feyel and A.S.  \"Ust\"unel,  {\it The strong solution of the Monge-Amp\'ere equation on the 
 Wiener space for log-concave measures: General case}. Jour. Func. Analysis,  {\bf 232}, (2006), 29-55.


 \bibitem{FU4} D. Feyel and A.S.  \"Ust\"unel,  {\it Log-concave measures}. arXiv preprint 1005.5127v1.

 
\bibitem{G} L. Gross,  {\it Logarithmic Sobolev inequalities}. Amer. J. Math. {\bf 97}, no. 4, (1975), 1061-1083.


 \bibitem{LMW} J. Lebowitz, P. Mounaix, W.-M. Wang, {\it Approach to equilibrium for the stochastic NLS},
Comm. Math. Phys. {\bf 321}, no. 1, (2013) 68-84.

\bibitem{N73} Nelson, E., The free Markoff field, J. Func. Anal. 12 (1973), 21 l-227.


 \bibitem{P} A.Pr\'ekopa, {\it On logarithmic concave measures and functions}. Acta Sci. Math. (Szeged) {\bf 33} (1973), 335-343.
 
 
  \bibitem{Sj} J. Sj\"ostrand, {\it Correlation asymptotics and Witten Laplacians},
Algebra and Analysis  {\bf 8(1)} (1996), 160-191.


  \bibitem{Sp} M. Spivak, {\it A Comprehensive Introduction to Differential Geometry, Vol. I,} Publish or Perish, Berkeley, 1970.
  
  
  
 \end{thebibliography}
\end{document}